%% file: main.tex
\pdfoutput=1 % <=== remove % for arxiv upload
\documentclass[conference]{IEEEtran}
\usepackage{cite}
\usepackage{amsmath,amssymb,amsthm,amsfonts,mathtools}
\usepackage{algorithm}
\usepackage{algpseudocode}
\usepackage{graphicx}
\usepackage{color}
\usepackage{textcomp}
\usepackage{xcolor}
\usepackage{comment} 
\usepackage{amsmath}
\usepackage{hyperref}
\usepackage[T1]{fontenc}
\usepackage{lmodern}
\usepackage{textcomp}
\usepackage{latexsym}
\usepackage[symbol]{footmisc}
\usepackage{tabularx}
\usepackage{booktabs}
\usepackage[caption=false]{subfig}

\newcommand{\argmin}{\mathop{\rm arg~min}\limits}

\newcommand{\fvec}{\boldsymbol{f}}
\newcommand{\wvec}{\boldsymbol{w}}

\newcommand{\lamvec}{\boldsymbol{\lambda}}
\def\BibTeX{{\rm B\kern-.05em{\sc i\kern-.025em b}\kern-.08em
    T\kern-.1667em\lower.7ex\hbox{E}\kern-.125emX}}

\newtheorem{proposition}{Proposition}
\newtheorem{theorem}{Theorem}

\algtext*{EndFor}
\algtext*{EndWhile}
\algtext*{EndIf}
\algtext*{EndProcedure}
\algtext*{EndFunction}

\newcommand{\Continue}{\textbf{continue}}

\begin{document}

\title{
 Preference-Optimal Multi-Metric Weighting for Parallel Coordinate Plots
}

\author{
    \IEEEauthorblockN{
        Chisa Mori\textsuperscript{*}~\textsuperscript{1}, Shuhei Watanabe\textsuperscript{*}~\textsuperscript{2,3}, Masaki Onishi~\textsuperscript{3}, Takayuki Itoh~\textsuperscript{1}
    } \vspace{1.5mm}
    \IEEEauthorblockA{
        \textsuperscript{*} Equal Contribution
    } \vspace{2mm}
    \IEEEauthorblockA{
        \textsuperscript{1}\textit{Ochanomizu University}, Tokyo, Japan \\
        \texttt{\{mori.chisa, itot\}@is.ocha.ac.jp} \\
    } \vspace{1.5mm}
    \IEEEauthorblockA{
        \textsuperscript{2}\textit{Preferred Networks Inc.}, Tokyo, Japan \\
        \texttt{shuheiwatanabe@preferred.jp}
    } \vspace{1.5mm}
    \IEEEauthorblockA{
        \textsuperscript{3}\textit{Advanced Industrial Science and Technology (AIST)}, Tokyo, Japan \\
        \texttt{onishi-masaki@aist.go.jp}
    }
}

\maketitle

\begin{abstract}
    Parallel coordinate plots (PCPs) are a prevalent method to interpret the relationship between the control parameters and metrics.
    PCPs deliver such an interpretation by color gradation based on a single metric.
    However, it is challenging to provide such a gradation when multiple metrics are present.
    Although a na\"ive approach involves calculating a single metric by linearly weighting each metric, such weighting is unclear for users.
    To address this problem, we first propose a principled formulation for calculating the optimal weight based on a specific preferred metric combination.
    Although users can simply select their preference from a two-dimensional (2D) plane for bi-metric problems, multi-metric problems require intuitive visualization to allow them to select their preference.
    We achieved this using various radar charts to visualize the metric trade-offs on the 2D plane reduced by UMAP.
    In the analysis using pedestrian flow guidance planning, our method identified unique patterns of control parameter importance for each user preference, highlighting the effectiveness of our method.
\end{abstract}


\input{intro}

\input{background-knowledge}

\input{method}

\input{experiments}

\input{related-work}

\input{conclusion}

\newpage

\bibliographystyle{plain}
\bibliography{myrefs}

\end{document}

%% file: intro.tex
\section{Introduction}

In the real world, decision makers often need to explain the relationship between the control parameters and their effects on the metrics of interest.
For example, human flow simulation, which attempts to mimic human flow in a real environment, has various control parameters that must be adapted by decision makers such that the real human flow becomes idealistic~\cite{yamashita2013crowdwalk}.
Decision making is particularly challenging when many control parameters exist.

Parallel coordinate plots (PCPs)~\cite{heinrich2013state} are among the most prevalent methods for interpreting the relationship between potentially high-dimensional control parameters and metrics of interest.
PCPs draw polylines where each polyline represents the control parameter values and their corresponding metric values.
In particular, each polyline is colored based on a metric value such that decision makers can distinguish the pattern in the performant control parameters.
However, practitioners often encounter multi-metric problems, which have multiple metrics to consider, and the coloring of PCPs is unclear owing to the lack of a unique metric, shown in Figure~\ref{fig:intro:tangential-2d-conceptual}.

\begin{figure}[tb]  
\begin{center}
\includegraphics[width=0.48\textwidth]{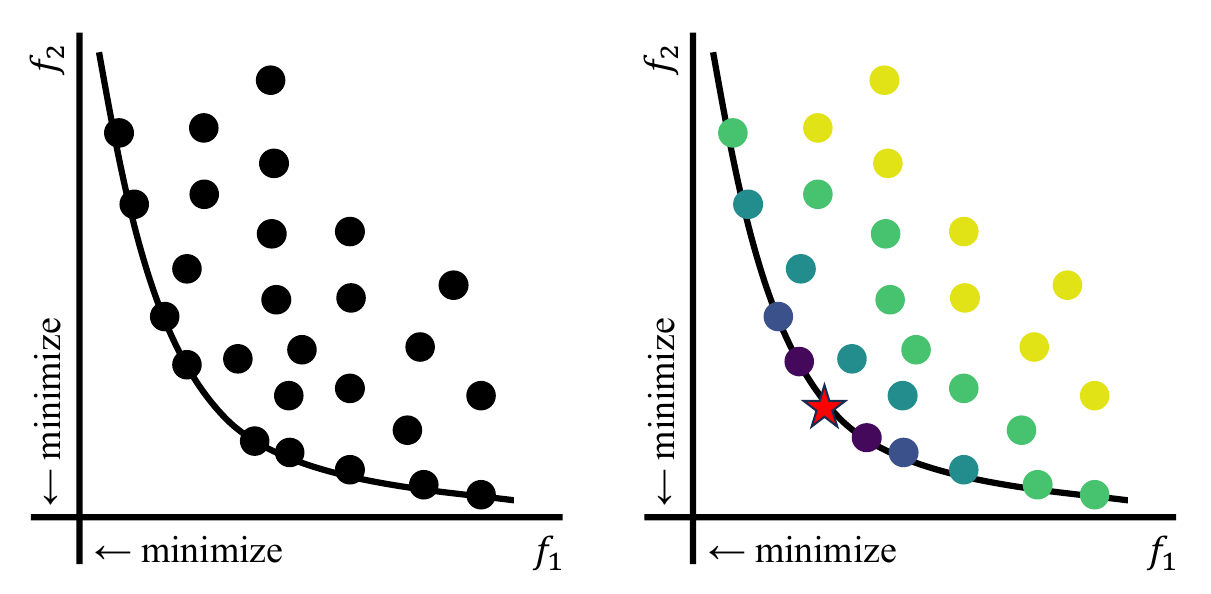}
\caption{
    Conceptual visualization of the preference-optimal weights in a bi-metric problem.
    For problems with three or more metrics, we provide additional visual support using radar charts as described in Section~\ref{main:section:method}.
    \textbf{Left}: A set of bi-metrics in the 2D plane.
    Black dots represent the observed bi-metrics, whereas the black line represents the approximated Pareto front.
    In principle, all the observations on the approximated Pareto front are equally performant considering multi-objective optimization.
    \textbf{Right}: Users provide their preference by selecting a point, e.g., the red star, on the approximated Pareto front.
    Our proposition, cf. Section~\ref{main:section:method}, enables calculating the optimal weights given the choice.
    For example, if the user selects the red star, each observation can be colored as in the right figure, where blue is better and yellow is worse under the user's choice.
    By doing so, we can rank the Pareto solutions based on the user preference.
}
\label{fig:intro:tangential-2d-conceptual}
\end{center}
\end{figure}

To this end, we propose a mathematical formulation to color PCPs for multi-metric problems based on a user preference.
Our formulation calculates the optimal weights to obtain the weighted average of each metric, allowing us to color PCPs based on a single weighted metric.
Although this formulation provides a principled way to calculate a single weighted metric in general problems, selecting an ideal trade-off between each metric is not trivial for users, particularly when more than two metrics are present.
To address this problem, we provide visual support combined with dimension reduction and radar charts.

Through experiments, we demonstrated the effectiveness of the proposed algorithm using pedestrian flow simulation data.
Our analysis identified the important control parameters for improving various trade-offs.
Importantly, different trade-offs exhibit different patterns of important control parameters, showing the effectiveness of our method.

In summary, our contributions are as follows:
\begin{enumerate}
    \item a novel formulation to calculate the optimal weights for a single weighted metric given user feedback,
    \item intuitive visual support to allow users to select their preferences even for multi-metric problems using radar charts and UMAP, and
    \item the demonstration using a pedestrian flow simulator that shows the effectiveness of our proposed method.
\end{enumerate}

%% file: background-knowledge.tex
\begin{figure}[tb]
    \begin{center}
        \includegraphics[width=0.48\textwidth]{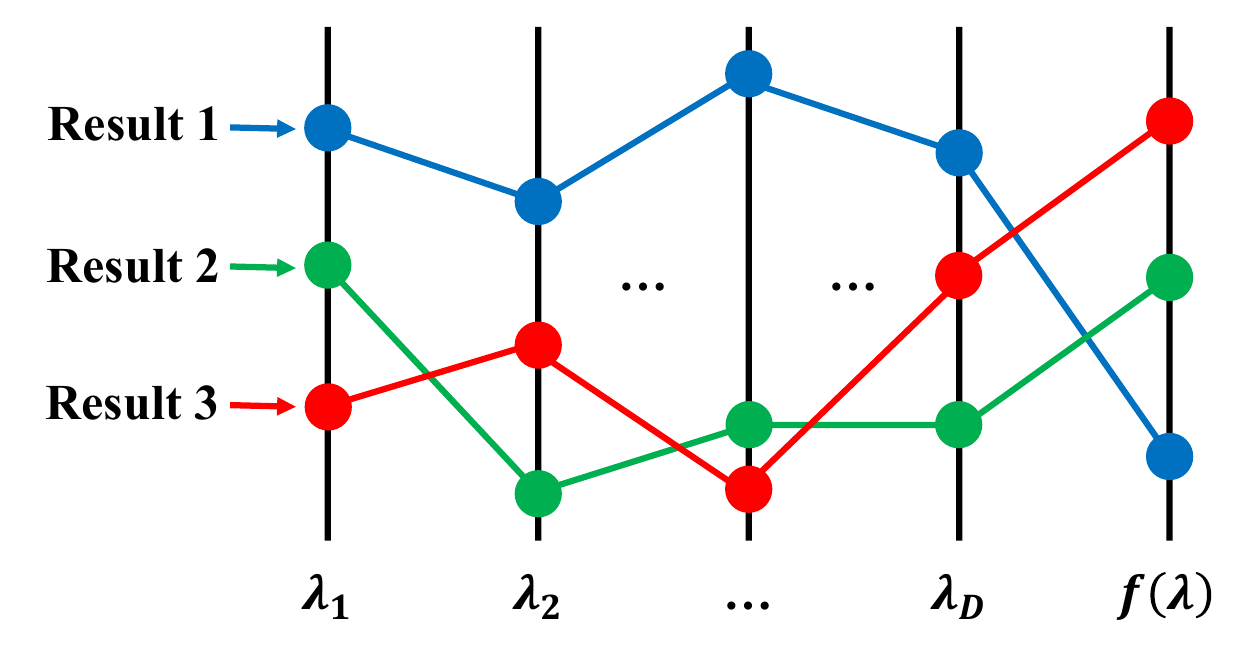}
        \caption{
            Example of PCP.
            Each polyline shows a control parameter vector used in an evaluation and its metric values.
            Each vertical axis represents either a control parameter or a metric. 
        }
        \label{fig:main:related-work:pcp}
    \end{center}
\end{figure}

\section{Background}

\subsection{Multi-Metric Problems}

A multi-metric problem is controlled by the parameter vector $\lamvec \coloneqq [\lambda_1, \cdots, \lambda_D]\in \Lambda \coloneqq \Lambda_1 \times \cdots \times \Lambda_D \subseteq \mathbb{R}^{D}$, where $\Lambda_d \subseteq \mathbb{R}$ denotes the domain of the $d$-th control parameter for $d \in [D] \coloneqq \{1, 2, \cdots, D\}$.
Examples of the control parameters include the ratio of pedestrians to be allocated to a certain route and the start time of pedestrian movement in a pedestrian flow simulation.
Each control parameter $\lamvec$ is evaluated using the metrics $\fvec(\lamvec) \coloneqq [f_1, \cdots, f_M]\in \mathbb{R}^M$ to \emph{minimize}, where $M$ denotes the number of metrics.
For example, we may evaluate a simulation with the total duration of an evacuation and how crowded each route.
This study aims to visualize some patterns of promising control parameters based on evaluated results, i.e., a set of control parameters and the corresponding metrics, $\{(\lamvec^{(n)}, \fvec^{(n)})\}_{n=1}^N$ where $N$ denotes the number of evaluated control parameter vectors and $\fvec^{(n)} \coloneqq \fvec(\lamvec^{(n)})$.

\subsection{Parallel Coordinate Plots (PCPs)}

Parallel coordinate plots (PCPs)~\cite{heinrich2013state} visualize each control parameter vector and the corresponding metric values as a polyline.
Figure~\ref{fig:main:related-work:pcp} illustrates PCPs.
Each polyline passes through each vertical axis, representing what value each control parameter and each metric had.
Since performant results are often of interest, we color each polyline differently depending on how performant each polyline is.
Although such colorization is straightforward for a single-metric problem, as PCPs can be simply colored by the metric, it is not trivial for a multi-metric problem.
This motivated our study.

\subsection{Pareto Solutions and Pareto Front}

In practice, Pareto optimality defines optimality in multi-metric problems.
Given a set of multiple metrics $\mathcal{F} \coloneqq \{\fvec^{(n)}\}_{n=1}^N$, the Pareto solutions of the metric set is $\mathcal{P} \coloneqq \{\fvec \in \mathcal{F} \mid \forall \fvec^\prime \in \mathcal{F}, \fvec^\prime \nprec \fvec\}$ where $\fvec^\prime \prec \fvec$ represents the dominance of $\fvec^{\prime}$ over $\fvec$, implying $f_m^\prime \leq f_m$ for all $m \in [M]$ and $f_m^\prime < f_m$ for some $m \in [M]$.
The Pareto front of $\mathcal{F}$ is the hypercurve where the Pareto solutions are supposed to be distributed.
Watanabe~\cite{watanabe2023python} provides further details.

%% file: method.tex
\section{Parallel Coordinate Plots for Multi-Metric Problems}
\label{main:section:method}

In this section, we first discuss PCPs for bi-metric problems and generalize them to higher-dimensional problems.
For both setups, we considered coloring PCPs based on user preference feedback and calculated a single weighted metric, i.e., a linear combination of each metric value, using the optimal weights computed from the user feedback.
Although user feedback can be easily provided by clicking on a preferred solution on the 2D plane, it is not trivial to achieve this for multi-metric problems, leading to our novel proposition in the combination of UMAP and radar charts.

\subsection{Preference Optimal Weight Calculation for Bi-Metric Problems}
In principle, PCPs require a single metric to color each polyline.
Therefore, we propose a method for calculating the weight of each metric.
Specifically, we color each polyline based on a single weighted metric $\phi_{\wvec}(\lamvec) \coloneqq \sum_{m=1}^M w_m f_m(\lambda)$ where $\wvec \coloneqq \{w_m\}_{m=1}^M$ denotes the weight vector subject to $\sum_{m=1}^M w_m = 1$.
We assume that the Pareto front $\mathcal{F}$ is represented as a strict convex function $g: [f_1^{\min}, f_1^{\max}] \rightarrow \mathbb{R}$ in the 2D plane where $f_1^{\min}, f_1^{\max}$ are the minimum and maximum of the first metric, and user feedback is provided in the form of $(f_1^u, g(f_1^u))$.
Thus, the following holds:
\begin{proposition}
    Assume that the optimal weights are defined such that a single metric $\phi_{\wvec^\star}(\lamvec)$ subject to $f_2 \geq g(f_1)$ achieves optimality at the user-specified point $(f_1^u, g(f_1^u))$, then the optimal weights are calculated as follows:
    \begin{equation}
        \wvec^\star = \biggl(-\frac{g^\prime(f_1^u)}{1 - g^\prime(f_1^u)}, \frac{1}{1 - g^\prime(f_1^u)}\biggr).
    \end{equation}
    \label{proposition:2d-optimal-weights}
\end{proposition}

\begin{figure}[tb]
    \begin{center}
    \includegraphics[width=0.48\textwidth]{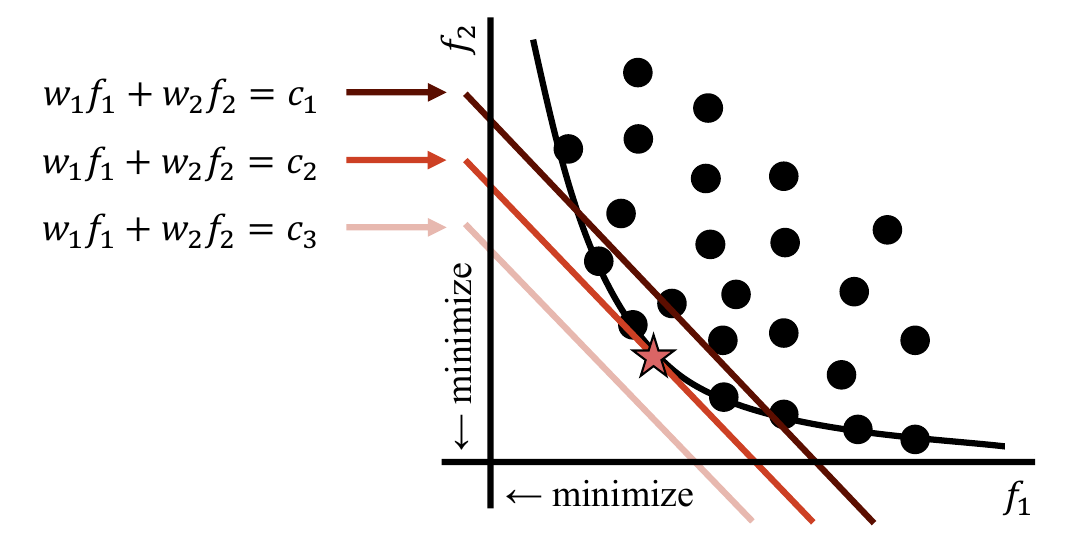}
    \caption{
        Conceptual visualization of Proposition~\ref{proposition:2d-optimal-weights}.
        The black dots represent each evaluated result, whereas the black line represents the approximated Pareto front $f_2 = g(f_1)$.
        The three different red lines represent the level sets, each with a different single weighted metric value.
        As shown by the red star, if the approximated Pareto front is strictly convex, the minimum possible $c$ is achieved when the level set is tangent to the approximated Pareto front.
        In our example, this is when we take $c_2$.
    }
    \label{fig:main:method:proof-2d-conceptual}
    \end{center}
\end{figure}

\begin{proof}
    Let $L_{\wvec}(c)$ be the level set $\{(f_1, f_2) | w_1 f_1 + w_2 f_2 = c\}$ and assume that $g^\prime$ is the derivative of $g$.
    Since $w_1 f_1 + w_2 f_2 = c$ can be transformed into $f_2 = h(f_1) = -\frac{w_1}{w_2}f_1 + \frac{c}{w_2}$ and $g(f_1) - h(f_1)$ is also a strict convex function, the minimum $c^\star$, which achieves $L_{\wvec}(c) \neq \emptyset$, is obtained when $g^\prime(f_1) - h^\prime(f_1) = 0$, i.e., $g^\prime(f_1) = -\frac{w_1}{w_2}$, and $\min_{f_1} g(f_1) - h(f_1) = 0$ hold.
    Owing to strict convexity, $g(f_1) - h(f_1)$ becomes minimum at the point where $g^\prime(f_1) - h^\prime(f_1) = 0$ holds.
    As the minimum must be at $(f_1^u, g(f_1^u))$, $g^\prime(f_1^u) = -\frac{w_1^\star}{w_2^\star}$ must hold.
    By solving $g^\prime(f_1^u) = -\frac{w_1^\star}{w_2^\star}$ and $w_1 + w_2 = 1$, we obtain $\wvec^\star = (-\frac{g^\prime(f_1^u)}{1 - g^\prime(f_1^u)}, \frac{1}{1 - g^\prime(f_1^u)})$.
    Note that Pareto optimality in 2D guarantees that the Pareto front $g(f_1)$ decreases monotonically, resulting in $g^\prime(f_1) \leq 0 \Rightarrow g^\prime(f_1) - 1 \neq 0$.
    Since the optimality at the user-specified point, i.e., $L_{\wvec^\star}(w_1^\star f^u_1 + w_2^\star f^u_2) = \{(f_1^u, f_2^u)\}$, holds, this completes the proof.
\end{proof}
An intuitive proof of this is provided in Figure~\ref{fig:main:method:proof-2d-conceptual}.
The figure shows three level sets $L_{\wvec}(c_1)$, which has two intersections with $f_2 \geq g(f_1)$, $L_{\wvec}(c_2)$, which is tangent to $f_2 \geq g(f_1)$, and $L_{\wvec}(c_3)$, and does not have any intersection with $f_2 \geq g(f_1)$ where $c_1 < c_2 < c_3$ as can be derived trivially from the hierarchical relationship of the intercepts.
The proof states that owing to strict convexity of $g$, there exists a level set that is tangent to $f_2 \geq g(f_1)$ and contains a unique point.

In this study, we use $g(f_1) = f_2 = \frac{b}{f_1 - a} + c$, where $a > f_1^{\min}$ and $b > 0$, for the approximation and fit the parameters $a, b, c$ by minimizing the squared error.

\subsection{Preference Optimal Weight Calculation for Multi-Metric Problems}

Given a user-specified point, the optimal weights can be calculated in a manner similar to that in Proposition~\ref{proposition:2d-optimal-weights}.
Let $g(\fvec) = b$ be the approximated Pareto front defined on $[f_1^{\min}, f_1^{\max}] \times \dots \times [f_M^{\min}, f_M^{\max}]$ and $\fvec^u \in g^{-1}(b)$ be the user-specified point.
Thus, the following holds:
\begin{theorem}
    The optimal weights are calculated as follows:
    \begin{equation}
        \wvec^\star \propto \nabla g(\fvec^u) \coloneqq \biggl(\frac{\partial g(\fvec^u)}{\partial f_1}, \dots, \frac{\partial g(\fvec^u)}{\partial f_M}\biggr).
        \label{main:method:optimal-weights-in-general}
    \end{equation}
\end{theorem}

\begin{proof}
    Similar to Proposition~\ref{proposition:2d-optimal-weights}, if $\fvec^u$ is on the approximated Pareto front, which is guaranteed from $\fvec^u \in g^{-1}(b)$, the tangential hyperplane at $\fvec^u$, i.e., $H \coloneqq \{\fvec | \nabla g(\fvec^u)^\top \cdot (\fvec - \fvec^u) = \boldsymbol{0} \}$, satisfies $H \cap \{\fvec \in g^{-1}(b)\} = \{\fvec^u\}$ owing to strict convexity.
    This completes the proof.
    Notably, the optimal weights are calculated by normalizing each element of $\nabla g(\fvec^u)$ with its summation.
\end{proof}

In our experiment, we use $g(\fvec) \coloneqq \prod_{m=1}^M (f_m - a_m) = b$, which is the general version of $g(f_1) = f_2 = b / (f_1 - a) + c$, leading to the following weights:
\begin{equation}
    \wvec^\star \propto \biggl(
        \frac{b}{f_1^u - a_1}, \dots, \frac{b}{f_M^u - a_M}
    \biggr),
\end{equation}
where we used the gradients of the approximated Pareto front at $\fvec^u$ calculated as follows:
\begin{equation}
    \frac{\partial g(\fvec^u)}{\partial f_i} = \prod_{m\neq i}^M (f_m^u - a_m) = \frac{b}{f_i^u - a_i}.
\end{equation}
Thus, we can color PCPs based on a single weighted metric, as in bi-metric problems.
Recall that although we chose $g(\fvec) = \prod_{m=1}^{M} (f_m - a_m) = b$ as the approximated Pareto front, the choice of the function can be arbitrary as long as the function is strictly convex and differentiable in the domain of $\fvec$.
However, it is difficult for users to select $\fvec^u$, unlike in bi-metric problems.
To this end, we propose visual support that enables users to intuitively select a trade-off of interest.

\begin{figure}[tb]
    \begin{center}
    \includegraphics[width=0.48\textwidth]{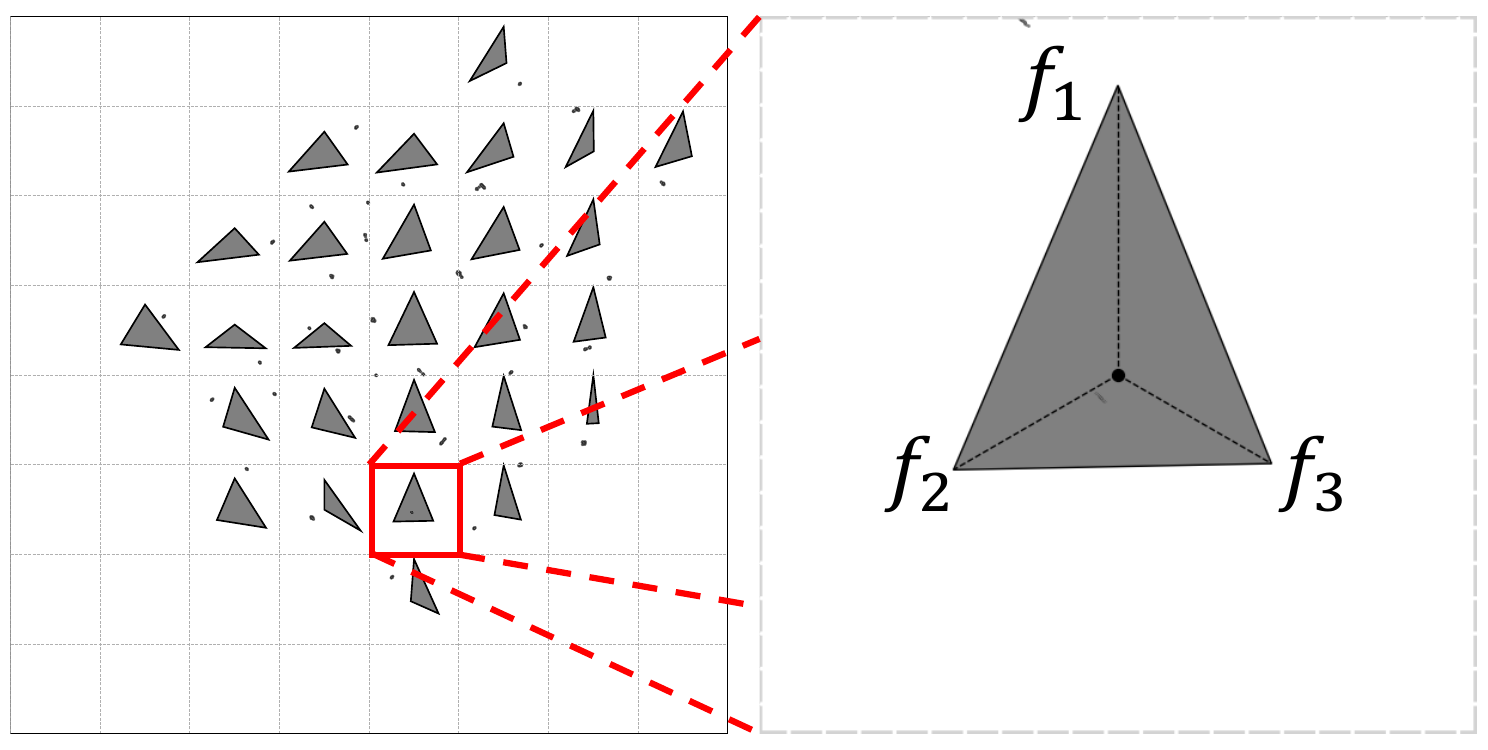}
    \caption{
        Visual support to select a multi-metric trade-off.
        \textbf{Left}: Pareto solutions in the space reduced by UMAP.
        Pareto solutions in each lattice are summarized by a radar chart.
        \textbf{Right}: Radar chart that shows the mean Pareto solutions in each lattice.
        Each radar chart shows how good each metric is.
        The larger a radar chart is, the better the Pareto solutions in the lattice are.
    }
    \label{fig:main:method:radarchart}
    \end{center}
\end{figure}

\begin{figure*}[tb]
    \begin{center}
    \includegraphics[width=0.75\textwidth]{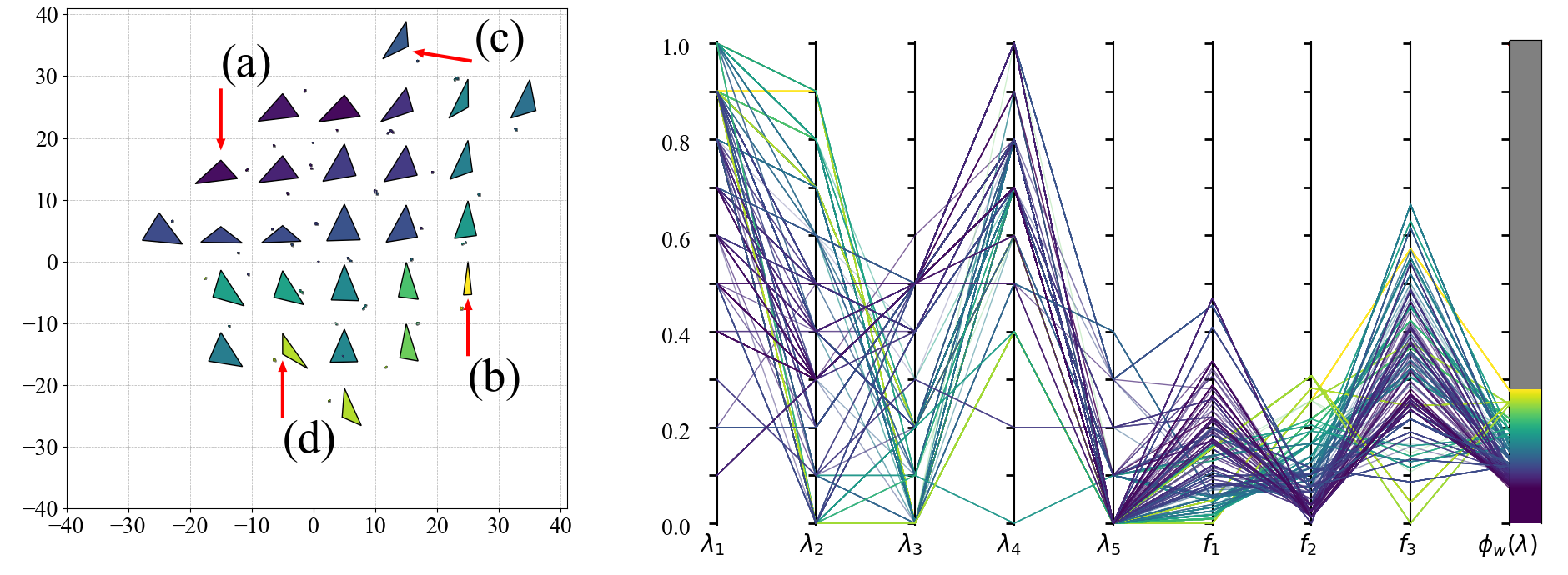}
    \caption{
        Illustration of the proposed method.
        \textbf{Left}: Radar charts of mean Pareto solutions on the 2D plane reduced by UMAP.
        Each radar chart shows the mean Pareto solution in the corresponding lattice, which helps users select their preference intuitively.
        UMAP enables radar charts with similar patterns to be located close to each other.
        The optimal weights are calculated based on the selected radar chart.
        \textbf{Right}: PCPs colored based on Radar Chart (a).
        Blue represents better single weighted metric values, whereas yellow represents worse single weighted metric values.
    }
    \label{fig:main:method:proposed-method}
    \end{center}
\end{figure*}

\subsection{Visual Support to Select Multi-Metric Trade-off}
An overview of our proposition is shown in Figure~\ref{fig:main:method:radarchart}.
To provide visual support, we first reduce the Pareto solutions $\mathcal{P}$ to a 2D plane, which is divided into lattices at even intervals, using UMAP~\cite{mcinnes2018umap}~\footnote{
\url{https://github.com/lmcinnes/umap}
}.
In particular, radar charts with similar patterns are located close to each other, making it easier for users to select a radar chart.
We then calculate the mean Pareto solution for each lattice.
Finally, we plot a radar chart for each lattice using the following:
\begin{equation}
        \bar{f}_m^\prime \coloneqq 1 - \frac{\bar{f}_m - f_m^{\min}}{f_m^{\max} - f_m^{\min}}.
    \label{main:method:normalization-of-metrics}
\end{equation}
Note that $\bar{f}_m$ is the $m$-th metric value of the mean Pareto solutions in a lattice, and $f_m^{\min}, f_m^{\max}$ are the minimum and maximum $m$-th metric values of the Pareto solutions, respectively.
The optimal weights are calculated based on the user's choice of a radar chart from the 2D plane, as shown in Figure~\ref{fig:main:method:radarchart}.
The pseudocode for the proposed visual support is described in Algorithm~\ref{main:method:visualization-support}.

After a radar chart is specified, we first calculate the nearest Pareto solution $\fvec^u$ on the approximated Pareto front $g(\fvec) \coloneqq \prod_{m=1}^M (f_m - a_m) = b$ using the sequential least squares programming in \texttt{scipy.minimize}~\footnote{
    \url{https://docs.scipy.org/doc/scipy/reference/optimize.minimize-slsqp.html}
}.
We assume $\fvec^r$ is the mean Pareto solution specified by the radar chart.
We then find the nearest Pareto solution by solving the following:
\begin{equation}
    \fvec^u \in \argmin_{\fvec \in g^{-1}(b)} \| \fvec^r - \fvec \|_{2}.
    \label{main:method:find-nearest-point-on-curve}
\end{equation}
The formulation above is equivalent to $\argmin_{\fvec \in \mathbb{R}^M} \| \fvec^r - \fvec \|_{2}$ subject to $g(\fvec) = b$.
Once the nearest Pareto solution is identified, the optimal weights can be computed easily using Eq.~(\ref{main:method:optimal-weights-in-general}).
Using the optimal weights, we color PCPs as in Figure~\ref{fig:main:method:proposed-method}.

\begin{algorithm}[tb]
  \caption{PCPs for Multi-Metric Problems}
  \label{main:method:visualization-support}
  \begin{algorithmic}[1]
    \Statex{\textbf{Input}: evaluated results $\mathcal{D} \coloneqq \{(\lamvec^{(n)}, \fvec^{(n)})\}_{n=1}^N$}
    \Statex{}
    \State{(1) \textbf{Initialization}}
    \State{Extract the Pareto solutions $\mathcal{P}$ from $\mathcal{D}$}
    \State{Approximate the Pareto front by fitting the parameters $\{a_m\}_{m=1}^M, b$ in $g(\fvec) = \prod_{m=1}^M (f_m - a_m) = b$}
    \State{Apply UMAP to $\mathcal{P}$}
    \For{each lattice}
    \If{No solutions in the lattice}
    \State{\Continue}
    \EndIf
    \State{Collect the solutions in the lattice.}
    \State{Calculate the mean Pareto solution $\bar{\fvec}$}
    \State{Transform $\bar{\fvec}$ into $\bar{\fvec}^\prime$ using Eq.~(\ref{main:method:normalization-of-metrics})}
    \State{Plot a radar chart in the lattice}
    \EndFor
    \Statex{}
    \State{(2) \textbf{PCPs with the optimal weights}}
    \State{Select a radar chart}
    \State{Solve Eq.~(\ref{main:method:find-nearest-point-on-curve}) to find the nearest Pareto solution $\fvec^u$}
    \State{Calculate the optimal weights $\wvec^\star$ based on Eq.~(\ref{main:method:optimal-weights-in-general})}
    \State{Color PCPs based on the optimal weights $\wvec^\star$}
  \end{algorithmic}
\end{algorithm}

%% file: experiments.tex
\begin{figure*}[tb]
  \begin{center}
    \subfloat[PCPs using the optimal weights by selecting Radar Chart (b) in Figure~\ref{fig:main:method:proposed-method}\label{subfig:pcp-b}]{
      \includegraphics[width=0.30\textwidth]{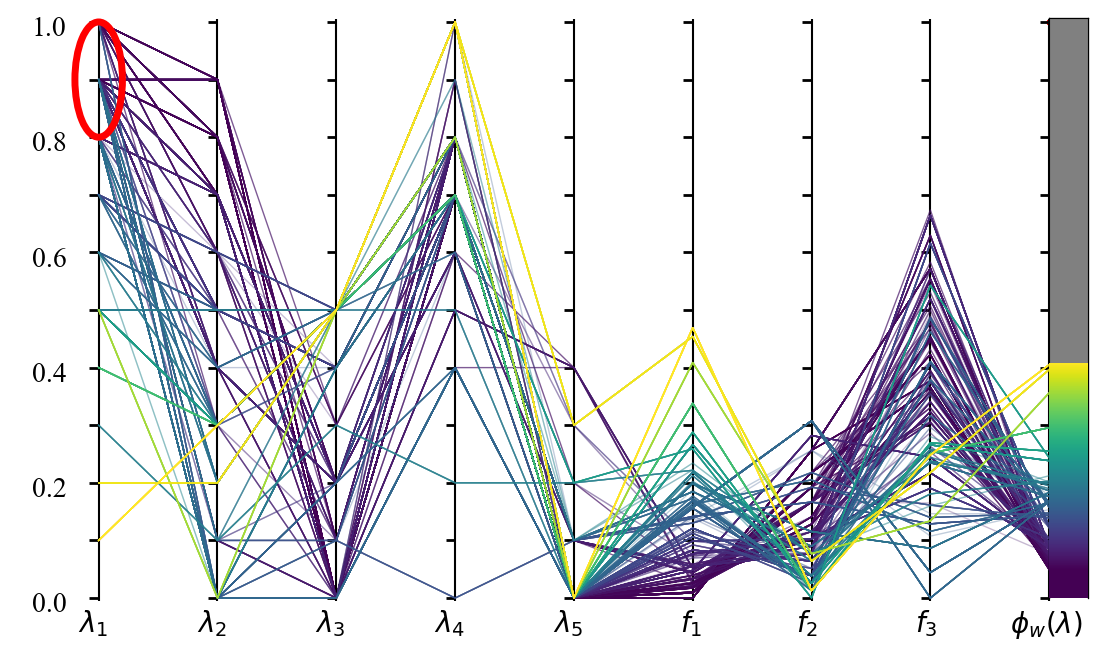}
    }
    \hspace{0.02\textwidth}
    \subfloat[PCPs using the optimal weights by selecting Radar Chart (c) in Figure~\ref{fig:main:method:proposed-method}\label{subfig:pcp-c}]{
      \includegraphics[width=0.30\textwidth]{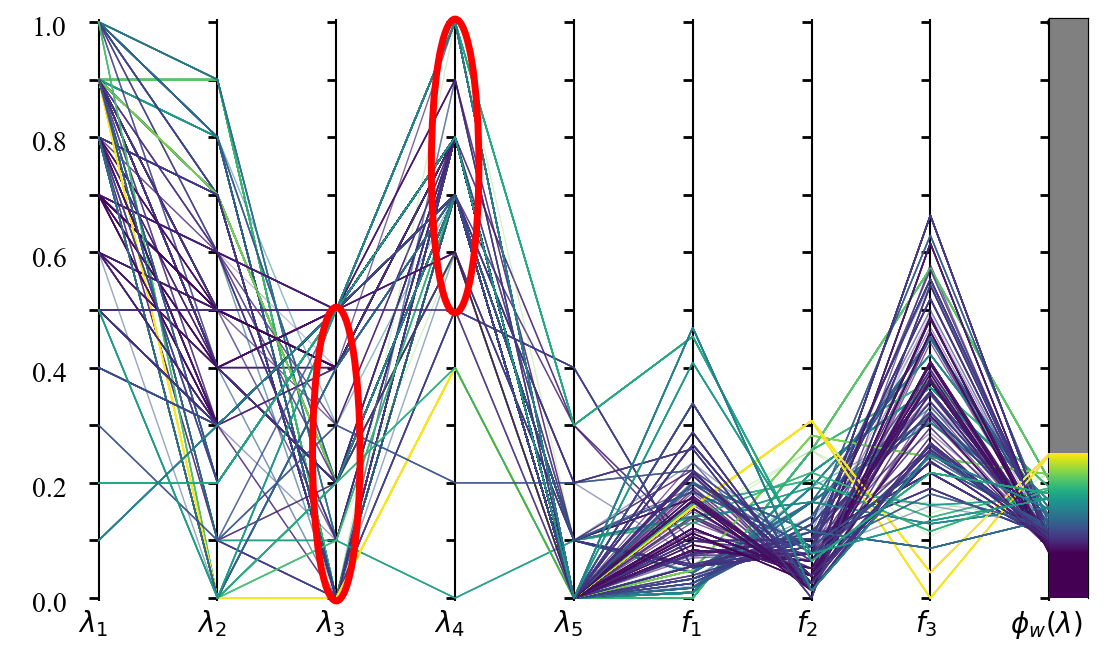}
    }
    \hspace{0.02\textwidth}
    \subfloat[PCPs using the optimal weights by selecting Radar Chart (d) in Figure~\ref{fig:main:method:proposed-method}\label{subfig:pcp-d}]{
      \includegraphics[width=0.30\textwidth]{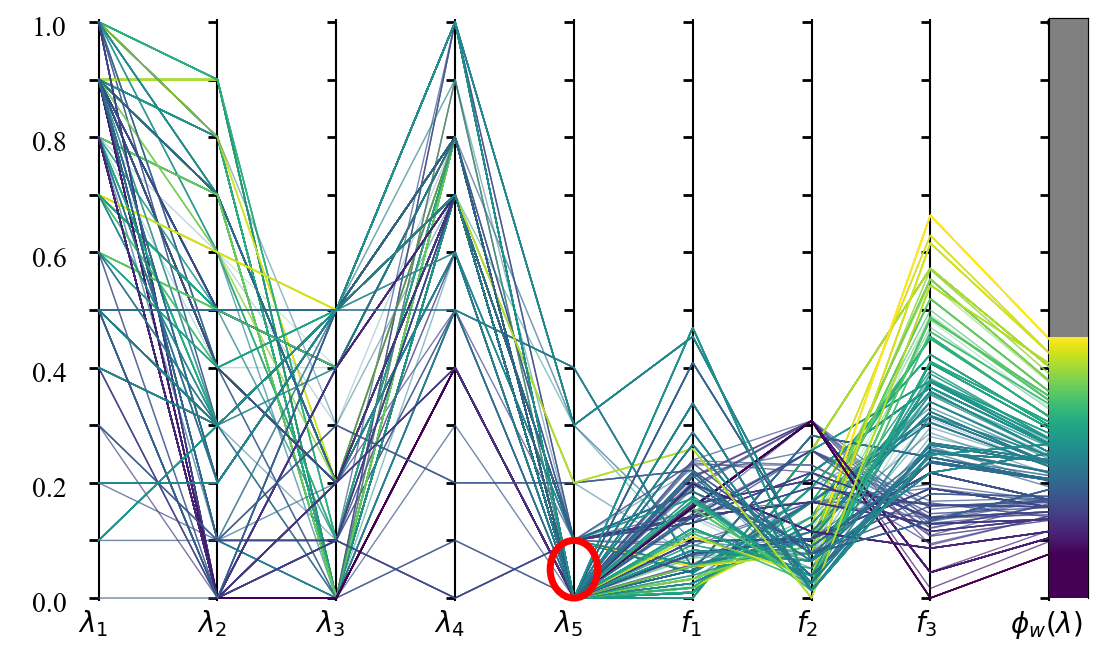}
    }
    \caption{
        PCPs based on three different radar charts.
        The PCPs are colored so that blue represents better single weighted metric values and yellow represents worse single weighted metric values.
        \textbf{Left}: The weights that put emphasis on $f_1$.
        A larger $\lambda_1$ improves $f_1$ as can be seen in the red circle. 
        \textbf{Center}: The weights that put emphasis on $f_2$.
        A smaller $\lambda_3$ or a larger $\lambda_4$ improves $f_2$ as can be seen in the red circles.
        \textbf{Right}: The weights that put emphasis on $f_3$.
        A smaller $\lambda_5$ improves $f_3$ as can be seen in the red circle.
    \label{fig:main:experiments:pcps-with-different-weights}
    }
  \end{center}
\end{figure*}

\section{Real-World Application to Pedestrian Flow Simulation}

\subsection{Problem Setup}
In this section, we provide an application example using CrowdWalk~\cite{yamashita2013crowdwalk}~\footnote{
\url{https://github.com/crest-cassia/CrowdWalk}
}, a pedestrian flow simulator in evacuation guidance.
The pedestrian flow in the simulator is controlled by five control parameters: 
$\lambda_1$, the time interval of the start of evacuation, $\lambda_2$ and $\lambda_3$, the percentages of guidance from the east side to the west side at two different locations, and $\lambda_4$ and $\lambda_5$, the percentages of guidance to the underground at two different locations, respectively.
Each control parameter was normalized to $[0, 1]$ and we discretized the domain of each control parameter into $\{k/10\}_{k=0}^{10}$, leading to $11^5$ control parameter combinations overall.
We consider three metrics in this problem:
$f_1$, the degree of congestion, $f_2$, the overall evacuation completion time, and $f_3$, the total traveling distance of pedestrians, where all the metrics were normalized to $[0, 1]$ for convenience.
Note that PCPs visualize the Pareto solutions and the top-30 observations with respect to the single weighted metric $\phi_{\wvec}(\lamvec)$, and we sampled 1000 control parameter vectors using \texttt{TPESampler}~\footnote{
    \url{https://optuna.readthedocs.io/en/stable/reference/samplers/generated/optuna.samplers.TPESampler.html}
} in Optuna~\cite{akiba2019optuna} and performed the simulator with these control parameters to collect the results.
We defer the details of the TPE algorithm to Watanabe~\cite{watanabe2023tree} and those of the multi-objective version to Ozaki \textit{et al.}~\cite{ozaki2020multiobjective,ozaki2022multiobjective}.

\subsection{Visualization Results}

Figure~\ref{fig:main:experiments:pcps-with-different-weights} visualizes the PCPs colored based on different metric trade-offs, which are indicated by (b), (c), and (d) in Figure~\ref{fig:main:method:proposed-method}.
Recall that a concentration of PCP lines with a similar color on a control parameter axis indicates that the control parameter is the key factor to achieve the single weighted metric values close to those specified by the color~\cite{watanabe2023ped}.

The PCPs in Figure~\ref{subfig:pcp-b} put more emphasis on congestion, $f_1$, by selecting Radar Chart (b) in Figure~\ref{fig:main:method:proposed-method}.
As marked by the red circle, high values in $\lambda_1$ attract purple lines.
It implies that a longer time interval is important at the start of evacuation.
By doing so, sufficient capacity will be preserved in the flow line, reducing congestion.
The PCPs in Figure~\ref{subfig:pcp-c} focus more on the overall evacuation completion time, $f_2$, by selecting Radar Chart (c) in Figure~\ref{fig:main:method:proposed-method}.
As shown by the red circles, purple lines concentrate at low values in $\lambda_3$ and high values in $\lambda_4$.
It indicates that east side guidance in $\lambda_3$ and underground guidance in $\lambda_4$ are essential to reduce the overall evacuation completion time.
Finally, we consider the single weighted metric with a higher weight on the total traveling distance of pedestrians, $f_3$, by selecting Radar Chart (d) in Figure~\ref{fig:main:method:proposed-method}.
As can be seen in the PCPs in Figure~\ref{subfig:pcp-d}, low $\lambda_5$ is the key factor of this scenario.
It means that the underground route should be avoided as much as possible at the location controlled by $\lambda_5$ if the total traveling distance of pedestrians is the primary objective.

%% file: related-work.tex
\section{Related Work \& Practical Consideration}
\label{main:section:related-work}

The interpretation methods for the relationship between the control parameters and metrics can be roughly divided into two main categories:
(1) Importance quantification of each control parameter, and (2) visualization of high-performance control parameter distribution. 

Importance quantification is typically used to reduce the complexity of the analysis via the control parameter selection.
For example, f-ANOVA~\cite{hutter2014efficient}, PED-ANOVA~\cite{watanabe2023ped}, and SHAP~\cite{lundberg2017unified} quantify the contribution of each control parameter to the variations of the metric of interest.
Since these methods are intended for single-metric problems, they cannot be immediately applied to multi-metric problems.
To the best of our knowledge, few methods address importance quantification for multi-metric problems.
Theodorakopoulos \textit{et al.}~\cite{theodorakopoulos2024hyperparameter} extended f-ANOVA to multi-metric setups.
Although they also focused on the linear weighting of each metric, cf. Section~\ref{main:section:method}, their focus was to compute the importance of each control parameter across different weight pairs, requiring users to select the desirable weights manually.
Meanwhile, we focused on deriving the optimal weights based on the user feedback.
The disadvantage of importance quantification is information loss regarding the locations of promising control parameters.
More specifically, this approach compresses the distribution information of control parameters and metric values into numerical values, making it difficult to perform a more accurate analysis.
This motivates the use of PCPs to interpret the patterns in high-performance control parameters more precisely. 

Another approach involves visualizing the distribution of promising control parameters.
In practice, scatter plot matrices~\cite{grinstein2001high} are widely used in this category.
However, as the number of figures increases quadratically, PCPs, which allow us to view high-dimensional data in a figure, become more advantageous than scatter plot matrices for high-dimensional problems.
Another example is landscape visualization using dimension reduction to a two-dimensional (2D) plane and a predictive model~\cite{lindauer2019towards,sass2022deepcave}.
PCPs, including variants such as density-based parallel coordinates~\cite{heinrich2013state}, are also classified into this category.
PCPs are frequently used in multi-metric problems~\cite{fleming2005many,li2017read}.
However, these studies did not discuss how to color each polyline to enhance interpretability regarding the trade-off between metrics.
The goal of this study is to provide a mathematical formulation to color PCPs based on user preferences and visual support to allow users to intuitively provide their preferences.

Practically speaking, these interpretation methods assist in reducing experiment design or black-box optimization problems to easier ones.
The concrete examples are control parameter selection by checking the importance and making some metrics into constraints.
The effectiveness of the former approach is reported by Watanabe \textit{et al.}~\cite{watanabe2022multi,watanabe2023speeding}, who used PED-ANOVA for this, and Florea and Andonie~\cite{florea2020weighted}, who used f-ANOVA for this.
The latter approach is particularly essential for many-objective problems ($M \geq 4$) because most evaluations achieve Pareto optimality in such a case owing to the curse of dimensionality.
Although this problem is tackled independently in the context of $\epsilon$ constrained-based approach~\cite{song2024dual}, our method is another effective approach to make some metrics into inequality constraints, as we can visually define appropriate thresholds for each metric in PCPs colored by our method.
By relaxing the problems, finding desirable trade-offs becomes viable in an efficient manner.

%% file: conclusion.tex
\section{Conclusion}

This paper proposes a method to color PCPs for multi-metric problems given user feedback.
Our method facilitates the user's selection of metric preferences by displaying the Pareto solutions in the 2D plane reduced by UMAP and metric radar charts at each lattice.
By coloring PCPs based on the user preference, a more intuitive interpretation of promising control parameters for the specified metric trade-off becomes possible.
To verify this, we conducted a visualization and analysis using simulation results of evacuation guidance.
In the analysis, we considered three different trade-offs.
Interestingly, different control parameters were important for each scenario.
Although there have not been any methods to color PCPs effectively in such a case, we could successfully address this problem and demonstrate the effectiveness through our application.
As discussed in Section~\ref{main:section:related-work}, our approach has the potential to contribute to relaxing many-objective problems into fewer-metric problems with some inequality constraints.
Such a study is a possible future research topic.

%% file: main.bbl
\begin{thebibliography}{10}

\bibitem{akiba2019optuna}
T.~Akiba, S.~Sano, T.~Yanase, T.~Ohta, and M.~Koyama.
\newblock {O}ptuna: A next-generation hyperparameter optimization framework.
\newblock In {\em International Conference on Knowledge Discovery \& Data
  Mining}, 2019.

\bibitem{fleming2005many}
PJ. Fleming, RC. Purshouse, and RJ. Lygoe.
\newblock Many-objective optimization: An engineering design perspective.
\newblock In {\em International conference on evolutionary multi-criterion
  optimization}, 2005.

\bibitem{florea2020weighted}
AC. Florea and R.~Andonie.
\newblock Weighted random search for hyperparameter optimization.
\newblock {\em arXiv:2004.01628}, 2020.

\bibitem{grinstein2001high}
G.~Grinstein, M.~Trutschl, and U.~Cvek.
\newblock High-dimensional visualizations.
\newblock In {\em Visual Data Mining Workshop at KDD}, 2001.

\bibitem{heinrich2013state}
J.~Heinrich and D.~Weiskopf.
\newblock State of the art of parallel coordinates.
\newblock {\em Eurographics}, 2013.

\bibitem{hutter2014efficient}
F.~Hutter, H.~Hoos, and K.~Leyton-Brown.
\newblock An efficient approach for assessing hyperparameter importance.
\newblock In {\em International conference on machine learning}, 2014.

\bibitem{li2017read}
M.~Li, L.~Zhen, and X.~Yao.
\newblock How to read many-objective solution sets in parallel coordinates.
\newblock {\em Computational Intelligence Magazine}, 12, 2017.

\bibitem{lindauer2019towards}
M.~Lindauer, M.~Feurer, K.~Eggensperger, A.~Biedenkapp, and F.~Hutter.
\newblock Towards assessing the impact of {B}ayesian optimization's own
  hyperparameters.
\newblock {\em arXiv:1908.06674}, 2019.

\bibitem{lundberg2017unified}
SM. Lundberg and SI. Lee.
\newblock A unified approach to interpreting model predictions.
\newblock {\em Advances in Neural Information Processing Systems}, 2017.

\bibitem{mcinnes2018umap}
L.~McInnes, J.~Healy, and J.~Melville.
\newblock {UMAP}: Uniform manifold approximation and projection for dimension
  reduction.
\newblock {\em arXiv:1802.03426}, 2018.

\bibitem{ozaki2022multiobjective}
Y.~Ozaki, Y.~Tanigaki, S.~Watanabe, M.~Nomura, and M.~Onishi.
\newblock Multiobjective tree-structured {P}arzen estimator.
\newblock {\em Journal of Artificial Intelligence Research}, 73, 2022.

\bibitem{ozaki2020multiobjective}
Y.~Ozaki, Y.~Tanigaki, S.~Watanabe, and M.~Onishi.
\newblock Multiobjective tree-structured {P}arzen estimator for computationally
  expensive optimization problems.
\newblock In {\em Genetic and Evolutionary Computation Conference}, 2020.

\bibitem{sass2022deepcave}
R.~Sass, E.~Bergman, A.~Biedenkapp, F.~Hutter, and M.~Lindauer.
\newblock {DeepCAVE}: An interactive analysis tool for automated machine
  learning.
\newblock {\em arXiv:2206.03493}, 2022.

\bibitem{song2024dual}
S.~Song, K.~Zhang, L.~Zhang, and N.~Wu.
\newblock A dual-population algorithm based on self-adaptive epsilon method for
  constrained multi-objective optimization.
\newblock {\em Information Sciences}, 655, 2024.

\bibitem{theodorakopoulos2024hyperparameter}
D.~Theodorakopoulos, F.~Stahl, and M.~Lindauer.
\newblock Hyperparameter importance analysis for multi-objective {AutoML}.
\newblock {\em arXiv:2405.07640}, 2024.

\bibitem{watanabe2023python}
S.~Watanabe.
\newblock {P}ython tool for visualizing variability of {P}areto fronts over
  multiple runs.
\newblock {\em arXiv:2305.08852}, 2023.

\bibitem{watanabe2023tree}
S.~Watanabe.
\newblock Tree-structured {P}arzen estimator: Understanding its algorithm
  components and their roles for better empirical performance.
\newblock {\em arXiv:2304.11127}, 2023.

\bibitem{watanabe2022multi}
S.~Watanabe, N.~Awad, M.~Onishi, and F.~Hutter.
\newblock Multi-objective tree-structured {Parzen} estimator meets
  meta-learning.
\newblock {\em Meta-learning Workshop at Advances in Neural Information
  Processing Systems}, 2022.

\bibitem{watanabe2023speeding}
S.~Watanabe, N.~Awad, M.~Onishi, and F.~Hutter.
\newblock Speeding up multi-objective hyperparameter optimization by task
  similarity-based meta-learning for the tree-structured {P}arzen estimator.
\newblock {\em arXiv:2212.06751}, 2023.

\bibitem{watanabe2023ped}
S.~Watanabe, A.~Bansal, and F.~Hutter.
\newblock {PED-ANOVA}: Efficiently quantifying hyperparameter importance in
  arbitrary subspaces.
\newblock {\em arXiv:2304.10255}, 2023.

\bibitem{yamashita2013crowdwalk}
T.~Yamashita, T.~Okada, and I.~Noda.
\newblock Implementation of simulation environment for exhaustive analysis of
  huge-scale pedestrian flow.
\newblock {\em Journal of Control, Measurement, and System Integration}, 6,
  2013.

\end{thebibliography}
